\documentclass{article}
\usepackage[noend, linesnumbered, noline, ruled]{algorithm2e}
\usepackage{tikz}
\usepackage{paralist, fullpage}
\usepackage{amsmath,amsthm}
\usepackage{amssymb}
\usepackage{authblk}
\usepackage{array}
\usepackage[utf8]{inputenc}
\usepackage[T1]{fontenc}
\usepackage{enumerate}
\usepackage{verbatim, xspace}
\usepackage{xcolor, colortbl}

\usepackage[hidelinks]{hyperref}
\usepackage{cleveref}

\usetikzlibrary{decorations.pathreplacing,decorations.pathmorphing}

\ifdefined\DEBUG

  \newcommand{\fab}[1]{\textcolor{red}{#1}}
  \newcommand{\tom}[1]{\textcolor{blue}{#1}}
  
  \def\rem#1{{\marginpar{\raggedright\scriptsize #1}}}
  \newcommand{\fabr}[1]{\rem{\textcolor{red}{$\bullet$ #1}}}
  \newcommand{\tomr}[1]{\rem{\textcolor{blue}{$\bullet$ #1}}}
  \newcommand{\micr}[1]{\rem{\textcolor{green}{$\bullet$ #1}}}
\else

  \newcommand{\fab}[1]{#1}
  \newcommand{\tom}[1]{#1}
    
  \newcommand{\fabr}[1]{}
  \newcommand{\tomr}[1]{}
  \newcommand{\micr}[1]{}

\fi

\providecommand{\eps}{\varepsilon}

\date{}

\newcommand{\preal}{\mathbb{R}_{\ge 0}}
\newcommand{\sm}{\setminus}
\newcommand{\E}{\mathbb{E}}

\newtheorem{theorem}{Theorem}
\newtheorem{lemma}{Lemma}

\makeatletter
\newtheorem*{rep@theorem}{\rep@title}
\newcommand{\newreptheorem}[2]{%
\newenvironment{rep#1}[1]{%
\def\rep@title{#2 \ref{##1}}%
\begin{rep@theorem}}%
{\end{rep@theorem}}
}
\makeatother

\newreptheorem{theorem}{Theorem}
\newreptheorem{lemma}{Lemma}
\newreptheorem{corollary}{Corollary}
\newreptheorem{fact}{Fact}

\newcommand{\RMAS}{\textsc{Restricted Maximum Acyclic Subgraph}\xspace}

\newcommand{\MAS}{\textsc{Maximum Acyclic Subgraph}\xspace}

\newcommand{\VP}{\textsc{Vertex Pricing}\xspace}

\title{An LP-Rounding $2\sqrt{2}$ Approximation for\\ Restricted Maximum Acyclic Subgraph%
\thanks{Partially supported by the ERC Starting Grant NEWNET 279352 and by Foundation for Polish Science
grant HOMING PLUS/2012-6/2.}}

\author[1]{Fabrizio Grandoni}
\author[2]{Tomasz Kociumaka}
\author[2]{Michał Włodarczyk}
 
 \affil[1]{IDSIA, University of Lugano, Switzerland}
 \affil[2]{Faculty of Mathematics, Informatics and Mechanics, University of Warsaw, Poland}

\begin{document}
\maketitle

\begin{abstract}
\noindent In the classical \MAS problem (MAS), given a directed-edge weighted graph, we are required to find an ordering of the nodes that maximizes the total weight of forward-directed edges. 
MAS admits a $2$-approximation, and this approximation is optimal under the Unique Game Conjecture.%
\tomr{We should be consistent choosing `$r$-approximation' or `$r$~approximation'}

In this paper we consider a generalization of MAS, the \RMAS problem (RMAS), where each node is associated with a list of integer labels, and we have to find a labeling of the nodes so as to maximize the weight of edges whose head label is larger than the tail label. The best known (almost trivial) approximation for RMAS is $4$.

The interest of RMAS is mostly due to its connections with the \VP problem (VP). In VP we are given an undirected graph with positive edge budgets. A feasible solution consists of an assignment of non-negative prices to the nodes. The profit for each edge $e$ is the sum of its endpoints prices if that sum is at most the budget of $e$, and zero otherwise. Our goal is to maximize the total profit. The best known approximation for VP, which works analogously to the mentioned approximation algorithm for RMAS, is $4$. Improving on that is a challenging open problem. On the other hand, the best known $2$ \fab{inapproximability result} is due to a reduction from a special case of RMAS.

In this paper we present an improved LP-rounding $2\sqrt{2}$ approximation for RMAS. Our result shows that, in order to prove a $4$ hardness of approximation result for VP (if possible), one should consider reductions from harder problems. Alternatively, our approach might suggest a different way to design approximation algorithms for VP.

\end{abstract}

\section{Introduction}

In the classical \textsc{Maximum Acyclic Subgraph} problem (MAS) we are given a directed graph
$G = (V, E)$, with edge weights $\{w_e\}_{e\in E}$,
and we look for an ordering of the nodes so as to maximize the total weight of forward-oriented edges. MAS admits a $2$-approximation, which is optimal under the Unique Games Conjecture (UGC) \cite{mas_hardness}.

In this paper we consider the following generalization of MAS. In the \RMAS problem (RMAS) we are given the same input as for MAS, plus a set $L_v$ of integer labels for each node $v$\footnote{We assume that the lists $L_v$ are given explicitly.}. 
Our goal is to find a labeling $\{\ell(v)\}_{v\in V}$ of the nodes, $\ell(v)\in L_v$, that maximizes the weight of edges going from a lower label to a higher one. 
In other words, the objective function is 
$$
\sum_{\fab{e=}(u,v)\in E \atop \ell(u) < \ell(v)} w_{\fab{e}}
$$
%
%
Note that it is allowed to assign the same label to multiple nodes. MAS is the special case of RMAS where $L_v=\{1,\ldots,|V|\}$ for all nodes $v$.\fabr{Maybe we should define $n=|V|$}
Clearly, the $2$-hardness of approximation for MAS extends to RMAS. The best known approximation ratio for RMAS is $4$, which is achieved with an almost trivial algorithm.

In this paper we investigate the approximability of RMAS, and  we present an improved $2\sqrt{2}$ approximation for the problem. Our result combines the trivial $4$ approximation algorithm with a novel LP-rounding algorithm.

\subsection{Related Work}

\fabr{Something more about MAS. How that the 2-apx work? Why does not generalize to RMAS? \tom{The intuitive argument is that
for MAS the optimum is always at least $\frac{W}{2}$.}}
Our interest in RMAS is motivated by the following \VP problem (VP): we are given an undirected (multi-) graph $G=(V,E)$, with positive edge budgets $\{B_e\}_{e\in E}$. Our goal is to assign a non-negative price $p(v)$ to each node $v$ so as to maximize the sum $p(u)+p(v)$ over the edges $e=\{u,v\}$ such that $p(u)+p(v)\leq B_e$. 
Khandekar et al. \cite{khandekar} proved that VP is $2$-hard to approximate (under UGC) via a reduction from a special case of RMAS. \fab{Their reduction exploits instances of RMAS where} labels are non-negative, each $L_v$ contains $0$, and $L_u\cap L_v=\{0\}$ for any distinct $u,v\in V$.
\fab{Note that} such instances \fab{still} generalize MAS.
\fabr{I meant more details about the reduction to VP. Why do you want to stress the special case of RMAS? The reductions are always special in some sense}

VP is APX-hard even on bipartite graphs \cite{ERRS09}. The best known approximation for VP is $4$ \cite{BB06ec}, and improving on that (if possible) is a well-known challenging open problem. Interestingly enough, the mentioned approximation is obtained with an algorithm analogous to the best-known $4$ approximation for RMAS. So it is natural to wonder whether RMAS and VP are equally hard to approximate. Our result suggests that RMAS might actually be an easier problem than VP. Alternatively, it might suggest a way to design improved approximation algorithms for VP (though generalizing our approach to VP does not seem easy). 

VP belongs to a broader family of \emph{pricing} problems, which recently attracted a lot of attention. In particular, in the (single-minded unlimited-supply) {\sc Item Pricing} problem (IT), we are given a (multi-) hyper-graph $G=(V,E)$ with hyper-edge positive budgets $\{B_e\}_{e\in E}$. We have to assign a non-negative price $p(v)$ to each node, so as to maximize the objective function $\sum_{e\in E: p(e)\leq B_e}p(e)$, where $p(e)=\sum_{v\in e}p(v)$. This problem admits a $O(\log m + \log n)$ approximation, where $n$ is the number of nodes and $m$ the number of hyper-edges \cite{GHKKKM05soda} (see also \cite{BK06soda} for a refinement of this result). On the negative side, Demaine et al. \cite{DFHS08} show that this problem is hard to approximate within $\log^\delta n$, for some constant $\delta>0$, assuming that $NP\not\subseteq BPTIME(2^{n^{\eps}})$ for some $\eps>0$.
Better approximation algorithms are known for the special case where the maximum size $k$ of any hyper-edge is small. 
In particular, an $O(k)$ approximation is given in \cite{BB06ec}. As recently shown \cite{CLN13}, the latter result is (asymptotically) essentially the best possible under the Exponential Time Hypothesis.

VP is the special case of IT where all hyper-edges have size precisely $2$. Another well-studied special case of IT is the {\sc Highway} problem. Here one is given a path $P$ on the node-set $V$, and hyper-edges are forced to induce subpaths of $P$.
This problem was shown to be weakly $NP$-hard by Briest and Krysta \cite{BK06soda}, and strongly $NP$-hard by Elbassioni, Raman, Ray, and Sitters \cite{ERRS09}. Balcan and Blum \cite{BB06ec} give an $O(\log n)$ approximation for the problem. Gamzu and Segev~\cite{GS10} improved the approximation factor to $O(\log n/\log\log n)$. Elbassioni, Sitters, and Zhang \cite{ESZ07esa} developed a QPTAS, exploiting the profiling technique introduced by Bansal et al.~\cite{BCES06stoc}. Finally, a PTAS was given by Grandoni and Rothvo\ss~\cite{GR11}.

The {\sc Tree Tollbooth} problem is a generalization of the {\sc Highway} problem where we are given a tree $T$ on the node-set $V$, and hyper-edges are forced to induce paths in $T$.     
An $O(\log n)$ approximation was developed in \cite{ERRS09}, which was slightly improved to $O(\log n/\log \log n)$ by Gamzu and Segev~\cite{GS10}. 
For the case of uniform budgets an $O(\log\log n)$ approximation was given by Cygan et al. \cite{CGLPS12}. {\sc Tree Tollbooth} is $APX$-hard \cite{GHKKKM05soda}.

\section{An Improved Approximation Algorithm for RMAS}
\label{sec:alg}

In this section we present our improved approximation algorithm for RMAS. 
In Section \ref{ssec:trivial} we revisit the folklore $4$-approximation for the problem, that is one of our building blocks. 
In Section \ref{ssec:lp} we present and analyze a novel LP-based algorithm. 
Finally, in Section \ref{ssec:conclusions} we discuss the derandomization of both algorithms, and conclude with our main result. 

In the following $W = \sum_{e\in E} w_e$ is the sum of all the weights. 
Observe that trivially $W$ is an upper bound on the profit of the optimum solution. 
For a given node $u\in V$, let $\ell_{u,\min}=\min\{\ell:\ell\in L_u\}$ and $\ell_{u,\max}=\max\{\ell:\ell\in L_u\}$. 
Without loss of generality, we can assume that, for any edge $e=(u,u')\in E$, one has $\ell_{u,\min}<\ell_{u',\max}$. Otherwise, $e$ can be filtered out without changing the value of the optimum solution.  

\subsection{A Simple Randomized Algorithm}
\label{ssec:trivial}

Consider the following simple algorithm for RMAS: \tom{i}ndependently for each $u\in V$, set $\ell(u)=\ell_{u,\min}$ with probability $\frac{1}{2}$, and $\ell(u)=\ell_{u,\max}$ otherwise. 
\begin{lemma}\label{lem:quarter}
The above algorithm computes a solution of expected profit at least $\frac{W}{4}$.
\end{lemma}
\begin{proof}
Fix an edge $e=(u,u')\in E$. By the initial filtering, $\ell_{u,\min}<\ell_{u',\max}$. Therefore, with probability at least $\frac{1}{4}$ one has $\ell(u)<\ell(u')$ and we benefit from $e$. By linearity of expectation the total expected profit is at least $\sum_{e} \frac{w_e}{4}=\frac{W}{4}$.
\end{proof}

The above result provides a $4$ approximation for RMAS. This analysis of the approximation ratio turns out to be tight, \fab{by exploiting a result in \cite{alon_directed} (see also \cite{khandekar})}.
\begin{lemma}
The above algorithm has approximation factor at least $4$.
\end{lemma}
\begin{proof}
Alon et al. \cite{alon_directed} constructed a family of acyclic directed graphs
with maximal directed cut of size \hbox{$\frac{m}{4} + o(m)$}
where $m$ is the number of edges.
We use such a graph as an RMAS instance with $L_v = \{1, \dots, |V|\}$ for every vertex $v$ and unit weights.
As the graph is acyclic it admits a topological order.
Setting positions in this order as labels gives us the optimal solution equal to $m$.
However, the algorithm from Lemma \ref{lem:quarter} would only assign labels 1 and $|V|$, 
giving solutions of value not exceeding the size of the largest directed cut.
\end{proof}

%
%
%

\subsection{An LP-Rounding Algorithm}
\label{ssec:lp}

Observe that the algorithm from Section \ref{ssec:trivial} provides an approximation strictly better than $4$ whenever the optimum solution is strictly cheaper than $W$. In this subsection we present a better algorithm for the opposite case. 

For an instance $D$ of RMAS, consider the following LP-relaxation $LP(D)$ of $D$. Let $L = \bigcup_{v \in V} L_v$. 
\begin{align}
\max\quad & \sum_{e=(u,u')\in E} \sum_{\ell<\ell'} w_{e}y_{uu'}(\ell,\ell') &\nonumber \\
 s.t.\quad & \sum_{\ell\in L} x_{u}(\ell)=1 & \forall\; u\in V\nonumber \\ 
\label{lp:forb} & x_{u}(\ell)=0 & \forall\;u\in V, \forall\; \ell\in L\sm L_u \\
\label{lp:sumy} & \sum_{\ell'\in L} y_{uu'}(\ell,\ell') = x_{u}(\ell) &  \forall\; u,u'\in V, \forall\; \ell\in L \\
\label{lp:sym} & y_{uu'}(\ell,\ell') = y_{u'u}(\ell',\ell) & \forall\; u,u'\in V, \forall\; \ell,\ell'\in L\\
 & x_{u}(\ell),\;y_{uu'}(\ell,\ell')\ge 0 & \forall\; u,u'\in V, \forall\; \ell,\ell'\in L\nonumber
\end{align}
In the above LP, variable $x_u(\ell)$ denotes whether a vertex $u$ has label 
$\ell$, and variable $y_{uu'}(\ell,\ell')$ denotes whether simultaneously
$u$ has label $\ell$ and $u'$ has label $\ell'$.
For the sake of presentation, we defined the variables $x_{u}(\ell)$ and $y_{uu'}(\ell,\ell')$ also for unfeasible label assignments. Constraint \eqref{lp:forb} guarantees that such variables are set to zero.

Consider the following natural randomized LP-rounding algorithm. Let $(x,y)$ be an optimal solution to $LP(D)$. 
Observe that for a fixed vertex $v$ variables $x_{u}(\ell)$, $\ell\in L_u$, define a probability distribution. 
We draw $\ell(u)$ from this distribution, independently for each $u\in V$.
Then $\ell(u)=\ell$ with probability $x_{u}(\ell)$. 
\begin{lemma}\label{lem:round}
The above algorithm computes a solution of expected cost at least $\frac{lp^2}{2W}$, where $lp$ is the value of the optimal \fab{fractional} solution to $LP(D)$.
\end{lemma}
In order to prove the above lemma, we need the following technical result.
\begin{lemma}\label{lem:matrix}
Let $A=[a_{ij}]$ be an $n\times n$ matrix with $a_{ij}\in \preal$. Let $r_i=\sum_{j}a_{ij}$ and $c_j = \sum_{\tom{i}} a_{ij}$ be the sum of entries
in the $i$-th row and $j$-th column, respectively.\fabr{Shouldn't we ask A to be symmetric in the claim???} Then
$$\sum_{i<j} r_ic_j \ge \frac{1}{2}\bigg(\sum_{i<j}a_{ij}\bigg)^2.$$
\end{lemma}
\begin{proof}
We use Iverson notation: $[\phi]$ is $1$ if $\phi$ is satisfied, and $0$ otherwise.
By symmetry \tom{between $(i,j)$ and $(i',j')$} we have
\begin{equation}\label{eq:1}\Bigg(\sum_{i<j}a_{ij}\Bigg)^2=\sum_{i,i',j,j'}[i<j][i'<j']a_{ij}a_{i'j'}\le
2\sum_{i,i',j,j'}[i<j][i'<j'][i\le i']a_{ij}a_{i'j'}.\end{equation}
Clearly $(i\le i' \wedge i'<j') \Rightarrow i<j'$, and consequently
\begin{align}
\sum_{i,i',j,j'}[i<j][i'<j'][i\le i']a_{ij}a_{i'j'} & \le \sum_{i,i',j,j'}[i'<j'][i\le i']a_{ij}a_{i'j'}\le 
\sum_{i,i',j,j'}[i<j']a_{ij}a_{i'j'}\nonumber \\
& = \sum_{i<j'}\bigg(\sum_{j}a_{ij}\bigg)\bigg(\sum_{i'}a_{i'j'}\bigg)
=\sum_{i<j'}r_ic_{j'}\label{eq:2}.
\end{align}
The claim follows by combining \eqref{eq:1} and \eqref{eq:2}.
\end{proof}
\begin{proof} {\em (of Lemma \ref{lem:round})}
For an edge $e=(u,u')\in E$\fab{,} define $p_e = \sum_{\ell<\ell'} x_{u}(\ell)x_{u'}(\ell')$ and $q_e = \sum_{\ell<\ell'} y_{u,u'}(\ell,\ell')$. Note that the expected profit from $e$ equals $p_ew_e$, while the profit of the LP solution for the same edge is $q_ew_e$. In particular, $lp=\sum_{e\in E}q_ew_e$.

For each $e=(u,u')$ we apply Lemma~\ref{lem:matrix} to the $|L|\times |L|$ matrix $A$
with $a_{ij}=y_{uu'}(i,j)$.
By \eqref{lp:sumy} the sum $r_i$ of the entries in the $i$-th row is equal to $x_{u}(i)$.
Moreover, combining \eqref{lp:sumy} and \eqref{lp:sym}, \fab{one has} that the sum $c_j$ of the entries in 
the $j$-th column is $x_{u'}(j)$.
We conclude that 
\begin{equation*}
p_e = \sum_{i<j} x_{u}(i)x_{u'}(j) = \sum_{i<j} r_ic_j \overset{\text{Lem.\ref{lem:matrix}}}{\geq} \tfrac{1}{2}\bigg(\sum_{i<j}a_{i,j}\bigg)^2=\tfrac{1}{2}\bigg(\sum_{i<j}y_{uu'}(i,j)\bigg)^2 = \tfrac{1}{2}q_e^2.
\end{equation*}
As function $f(x)=x^2$ is convex, by Jensen's inequality with coefficients $\frac{w_e}{W}$
we obtain that the expected profit $\sum_{e}w_ep_e$ of the approximate solution satisfies:
\begin{equation*}
\sum_{e}w_ep_e \ge \tfrac{W}{2}\sum_e \tfrac{w_e}{W}q_e^2  \ge \tfrac{W}{2}\bigg(\sum_e \tfrac{w_e}{W}q_e\bigg)^2
= \frac{lp^2}{2W}.\qedhere
\end{equation*}
\end{proof}

\subsection{Derandomization and Conclusions}
\label{ssec:conclusions}

We start by observing that both the mentioned algorithms can be easily derandomized by using the method of conditional expectations. We next shortly describe how to do that. 
\begin{lemma}\label{lem:quarter_der}
The algorithm from Lemma \ref{lem:quarter} can be derandomized.
\end{lemma}
\begin{proof}
Let $Z$ be a random variable equal to the value of the integer solution computed by the randomized algorithm. We consider nodes in an arbitrary order $v_1,v_2,\ldots,v_{|V|}$. At each iteration $i=1,\ldots,|V|$, we have already fixed labels $\ell_j$ for each node $v_j$, $j<i$, such that the invariant  $\E[Z\vert \ell(v_j)=\ell_j, j=1,\ldots,i-1]\geq \E[Z]$ holds. In the considered iteration we fix the label $\ell_i$ for node $v_i$ as follows. We compute the two quantities $\E[Z\vert \ell(v_j)=\ell_j, j=1,\ldots,i-1\text{ and }\ell(v_i)=\ell_{v_i,\min}]$ and $\E[Z\vert \ell(v_j)=\ell_j, j=1,\ldots,i-1\text{ and }\ell(v_i)=\ell_{v_i,\max}]$.
Note that these quantities can be easily computed in polynomial time. Observe also that at least one of the two quantities is lower bounded by $\E[Z\vert \ell(v_j)=\ell_j, j=1,\ldots,i-1]$, hence by $\E[Z]$ because of the invariant. We set $\ell_i$ to the label in $\{\ell_{v_i,\min},\ell_{v_i,\max}\}$ that achieves the larger conditional expectation. It follows that the resulting deterministic algorithm computes a solution of profit at least $\E[Z]\geq \frac{W}{4}$.
\end{proof}
\begin{lemma}\label{lem:derand}
The algorithm from Lemma~\ref{lem:round} can be derandomized.
\end{lemma}
\begin{proof}
We use the method of conditional expectation similarly to the proof of Lemma \ref{lem:quarter_der}. Let $Z$ be a random variable equal to the value of the solution computed by the randomized algorithm. Recall from the proof of Lemma~\ref{lem:round} that
\begin{equation}\label{eq:prob}
\E[Z] = \sum_{e \in E} p_ew_e, \quad
\text{where }\quad p_e = \sum_{\ell<\ell'} x_{u}(\ell)x_{u'}(\ell').
\end{equation}
Let us choose some vertex $v_1$.
Observe that for some $\ell_1\in L_{v_1}$ it must be $\E[Z\,|\,l(v_1)=\ell_1] \ge \E[Z]$.
To compute $\E[Z\,|\,\ell(v_1)=\ell]$ we can set $x_{v_1}(\ell) = 1$ and
$x_{v_1}(\ell') = 0$ for $\ell' \neq \ell$ and use \eqref{eq:prob}.
Therefore such $\ell_1$ may be computed in polynomial time. 

We fix $\ell(v_1) = \ell_1$ and repeat this procedure on the remaining nodes considered in any order $v_2,\ldots,v_{|V|}$ until a label $\ell_i$ is chosen for each node $v_i$. The conditional expected value never decreases
so the value of the resulting solution is at least $\E[Z]\geq \frac{lp^2}{2W}$.
\end{proof}

We now have all the ingredients to prove the main result in this paper.
\begin{theorem}
RMAS admits a deterministic $2\sqrt{2}$-approximation algorithm.
\end{theorem}
\begin{proof}
Let $opt$ be the value of the optimal solution, $opt\leq lp \leq W$. Consider the algorithm which returns the better solution among the ones computed by the algorithms from Lemmas \ref{lem:quarter_der} and \ref{lem:derand}. The profit of the constructed solution is bounded from below by
$$
\max\Big\{\tfrac{W}{4},\tfrac{lp^2}{2W}\Big\}\ge \max\Big\{\tfrac{W}{4},\tfrac{opt^2}{2W}\Big\}=opt\cdot \max\Big\{\tfrac{W}{4opt},\tfrac{opt}{2W}\Big\}.
$$
The worst-case approximation factor is therefore $2\sqrt{2}$, which is achieved for $\frac{W}{opt}=\sqrt{2}$.
\end{proof}

\bibliographystyle{plain}
\bibliography{rmas}

\end{document}